\documentclass[runningheads]{llncs}
\usepackage[T1]{fontenc}

\usepackage{amssymb, dsfont}
\usepackage{paralist}
\usepackage{mathrsfs} 
\usepackage[dvipsnames]{xcolor} 
\usepackage{tikz}
\usepackage{graphicx}
\usepackage{pgfplots}
\usetikzlibrary{decorations} 
\usepackage{aligned-overset} 
\usepackage[ruled,vlined]{algorithm2e} 
\SetKwComment{Comment}{// }{} 
\usepackage{graphicx}
\usepackage{hyperref}
\usepackage{cleveref}
\usepackage{color}

\crefformat{lemma}{Lemma #2#1#3}
\crefformat{theorem}{Theorem #2#1#3}
\crefformat{proposition}{\mbox{Proposition #2#1#3}}
\crefformat{remark}{Remark #2#1#3}
\crefformat{equation}{(#2#1#3)}
\crefformat{corollary}{Corollary #2#1#3}
\crefformat{example}{Example #2#1#3}
\crefformat{definition}{Definition #2#1#3}
\crefformat{observation}{Observation #2#1#3}

\pgfplotsset{compat=1.14}
\usetikzlibrary{arrows.meta,arrows}
\usetikzlibrary{positioning,arrows}
\usetikzlibrary{patterns}
\usepgfplotslibrary{fillbetween}
\pgfplotsset{every x tick label/.append style={font=\small, yshift=0.0ex}}
\pgfplotsset{every y tick label/.append style={font=\small, yshift=0.0ex}}
\usetikzlibrary{patterns}
\makeatletter
\pgfdeclarepatternformonly[\LineSpace]{my north east lines}{\pgfqpoint{-1pt}{-1pt}}{\pgfqpoint{\LineSpace}{\LineSpace}}{\pgfqpoint{\LineSpace}{\LineSpace}}%
{
	\pgfsetcolor{\tikz@pattern@color}
	\pgfsetlinewidth{0.4pt}
	\pgfpathmoveto{\pgfqpoint{0pt}{0pt}}
	\pgfpathlineto{\pgfqpoint{\LineSpace + 0.1pt}{\LineSpace + 0.1pt}}
	\pgfusepath{stroke}
}
\makeatother
\newdimen\LineSpace
\tikzset{
	line space/.code={\LineSpace=#1},
	line space=10pt
}
\usepgfplotslibrary{fillbetween}
\tikzset{
	schraffiert/.style={pattern=horizontal lines,pattern color=#1},
	schraffiert/.default=black
}
\tikzstyle{densely dashed}=          [dash pattern=on 6pt off 2pt]
\tikzstyle{densely shadow}=          [dash pattern=on 6.5pt off 1.5pt]

\newcommand{\N}{\mathbb{N}}
\newcommand{\R}{\mathbb{R}}

\newcommand{\dist}{d}

\newcommand{\opt}{\textsc{Opt}}
\newcommand{\deliver}{\textsc{deliver\_and\_return}}

\newcommand{\waituntil}{\textsc{wait\_until}}

\newcommand{\schedule}{\textsc{follow\_schedule}}

\newcommand{\wait}{\textsc{Lazy}}
\renewcommand{\o}{O} 

\renewcommand{\epsilon}{\varepsilon}

\begin{document}
\title{An Improved Algorithm for Open Online Dial-a-Ride\thanks{Supported by DFG grant DI 2041/2.}}
%
%
\author{Júlia Baligács\inst{1}\orcidID{0000-0003-2654-149X} \and
Yann Disser\inst{1}\orcidID{0000-0002-2085-0454} \and 
{Nils~Mosis}\inst{1}\orcidID{0000-0002-0692-0647} \and
David Weckbecker\inst{1}\orcidID{0000-0003-3381-058X}}
\authorrunning{J. Baligács et al.}
%
\institute{TU Darmstadt, Germany \\
\email{\{baligacs,disser,mosis,weckbecker\}@mathematik.tu-darmstadt.de}}
\maketitle              
\begin{abstract}
We consider the open online dial-a-ride problem, where transportation requests appear online in a metric space and need to be served by a single server.
The objective is to minimize the completion time until all requests have been served.
We present a new, parameterized algorithm for this problem and prove that it attains a competitive ratio of $1 + \varphi \approx 2.618$ for some choice of its parameter, where $\varphi$ is the golden ratio.
This improves the best known bounds for open online dial-a-ride both for general metric spaces as well as for the real line.
We also give a lower bound of~$2.457$ for the competitive ratio of our algorithm for any parameter choice.

\end{abstract}

\section{Introduction}

In the online dial-a-ride problem, transportation requests appear over time in a metric space $(M,d)$ and need to be transported by a single server. 
Each request is of the form $r=(a,b;t)$, appears
at its starting position $a\in M$ at its release time $t\geq0$, and
needs to be transported to its destination $b\in M$. The server starts
at a distinguished point $O\in M$, called the origin, can move at
unit speed, and has a capacity $c\in\mathbb{N}\cup\{\infty\}$ that
bounds the number of requests it is able to carry simultaneously.
Importantly, the server only learns about request~$r$ when it appears
at time $t$ during the execution of the server's algorithm. 
Moreover, the total number of requests is initially unknown and the server
cannot tell upon arrival of a request whether it is the last one.\footnote{If the server can distinguish the last request, it can start an optimal schedule once all requests are released, achieving a completion time of at most twice the optimum.}
Requests do not have to be served in the same order in which they appear.

The objective of the open dial-a-ride problem is to minimize the time
until all requests have been served\emph{,} by loading each request
$r=(a,b;t)$ at point $a$ no earlier than time $t$, transporting
it to point $b$ and unloading it there. We consider the non-preemptive
variant of the problem, meaning that requests may only be unloaded
at their respective destinations. Note that, in contrast to the closed
variant of the problem, we do not require the server to return to
the origin after serving the last request.

As usual, we measure the quality of a (deterministic) online algorithm
in terms of competitive analysis. That is, we compare the completion
time $\textsc{Alg}(\sigma)$ of the algorithm to an offline
optimum completion time $\textsc{Opt}(\sigma)$ over all request sequences~$\sigma$.
Here, the offline optimum is given by the best possible completion
time that can be achieved if all requests are known (but not released)
from the start. The (strict) competitive ratio of the algorithm is
given by $\rho:=\sup_{\sigma}\textsc{Alg}(\sigma)/\textsc{Opt}(\sigma)$.\footnote{We adopt a strict definition of the competitive ratio that requires
a bounded ratio for all request sequences, i.e., we do not allow an
additive constant.} Note that, in particular, the running time of an algorithm does not play a role in its competitive analysis. 

\paragraph*{\textbf{our results.}}

We present a parameterized online algorithm $\textsc{Lazy}(\alpha)$
and show that it improves on the best known upper bound for open online
dial-a-ride for $\alpha=\varphi$, where $\smash{\varphi=\frac{1+\sqrt{5}}{2}}$
denotes the golden ratio. We also show a lower bound on potential
improvements for other parameter choices. More precisely, we show
the following.
\begin{theorem}\label{thm:CR_of_lazy}
$\textsc{Lazy}(\varphi)$ has competitive ratio $1+\varphi\approx2.618$
for the open online dial-a-ride problem on general metric spaces for
any server capacity $c\in\mathbb{N}\cup\{\infty\}$. For every $\alpha\geq1$ and any~$c\in\mathbb{N}\cup\{\infty\}$,
$\textsc{Lazy}(\alpha)$ has competitive ratio at least $\max\{1+\alpha,2+2/(3\alpha)\}$, even if the metric space is the real line.
\end{theorem}
In particular, we obtain a lower bound on the competitive ratio of
our algorithm, independent of~$\alpha$.
\begin{corollary}\label{cor:global_lower_bound}
For every $\alpha\geq0$, $\textsc{Lazy}(\alpha)$ has competitive
ratio at least \linebreak ${3/2+\sqrt{11/12}\approx2.457}$ for open online dial-a-ride on the line.
\end{corollary}
Our lower bound also narrows the range of parameter choices that could
allow improved competitive ratios.
\begin{corollary}\label{cor:phi+1_lower_bound}
For $\alpha\notin(2\varphi/3,\varphi)\approx(1.078,1.618)$, $\textsc{Lazy}(\alpha)$
has competitive ratio at least $1+\varphi$ for open online dial-a-ride on the line.
\end{corollary}
Our upper bound improves the best known upper bound
of~2.70 for general metric spaces~\cite{Birx/20}, and even the best
known upper bound of 2.67 for the real line~\cite{BirxDisser/19b}.
Figure~\ref{fig:state-of-art} gives an overview over the previous
upper bounds for open online dial-a-ride. Note that, in contrast to
previous results, $\textsc{Lazy}(\alpha)$ is not a so-called schedule-based
algorithm as defined in~\cite{Birx/20}, because it interrupts schedules.

We note that an upper bound of $\varphi+1\approx2.618$ was already
claimed in~\cite{Lipmann/03} for the Wait-or-Ignore algorithm, but the proof in~\cite{Lipmann/03} is inconclusive. 
While the general idea of our algorithm is similar to
Wait-or-Ignore, our implementation is more
involved and avoids issues in the analysis that are not being addressed
in~\cite{Lipmann/03}.
In particular, Wait-or-Ignore only waits at the origin, while our algorithm crucially also waits at other locations.

\paragraph*{\textbf{related work.}}

As listed in Figure~\ref{fig:state-of-art}, the best previously
known upper bound for open online dial-a-ride of~2.70 was shown by
Birx~\cite{Birx/20} and a slightly better bound of 2.67 for the
line was shown by Birx et al.~\cite{BirxDisser/19b}. In this paper,
we improve both bounds to $1+\varphi\approx2.618$. A better upper
bound of $1+\sqrt{2}\approx2.41$ is known for the preemptive variant
of the problem, due to Bjelde et al.~\cite{BjeldeDisserHackfeldEtal/20}.
The TSP problem is an important special case of dial-a-ride, where
$a=b$ for every request $(a,b;t)$, i.e., requests just need to be
visited. Bonifaci and Stougie gave an upper bound for open online
TSP on general metric spaces of $2.41$. Bjelde et al.~\cite{BjeldeDisserHackfeldEtal/20}
were able to show a tight bound of 2.04 for open online TSP on the
line. Birx et al.~\cite{BirxDisser/19b} showed that open online
dial-a-ride is strictly more difficult than open online TSP by providing
a slightly larger lower bound of 2.05. Weaker lower bounds for the
half-line were given by Lipmann~\cite{Lipmann/03}.
%
%
%
%
%

\begin{figure}
\begin{centering}
\begin{tikzpicture}[xscale=8]
  \node (O) at (2, 0) {};
  \node (L) at (3.4, 0) {};
  \node[rotate=90] (M) at (2.3,0) {$\boldsymbol\approx$};

 \node[align=right,execute at begin node=\setlength{\baselineskip}{10px}] (LB) at (2.22, 0.5) {lower\\bounds};
  \node[align=left,execute at begin node=\setlength{\baselineskip}{10px}] (UB) at (2.38,0.46) {upper\\bounds};

  \tikzstyle{tick}=[rectangle, thick, inner sep=0pt, minimum height=4pt, minimum width=0pt, draw]
  \node[tick, label={[label distance=.1cm]above:$\hspace{-.5cm} (2.04)\,\text{\cite{BjeldeDisserHackfeldEtal/20}}$}] (x204) at (2.03,0.06) {};
  \node[tick, label={[label distance=.1cm]below:$\hspace{.5cm}(2.05)$\,\cite{BirxDisser/19b}}] (x205) at (2.05,-0.06) {};

  \node[tick, label={[label distance=0.1cm]above:$3.41$\,\cite{Krumke0}}] (y341) at (3.31, 0.06) {};
  \node[tick, label={[label distance=0.1cm]below:($2.94$)\,\cite{BirxDisser/20}}] (y294) at (2.94, -.06) {};
  \node[tick, label=above:{$\hspace{.3cm}2.70$\,\cite{Birx/20}}] (y270) at (2.7, 0.06) {};
  \node[tick, label={[label distance=0.1cm]below:($2.67$)\,\cite{BirxDisser/19b}}] (y267) at (2.67, -0.06) {};

  \node[tick, very thick, label={[label distance=0.08cm]above:{$\hspace{-.3cm}\textbf{2.618}$}}] (y2618) at (2.618, 0.06) {};

  \draw[thick] (O) -> (M);
  \draw[->, thick] (M) -> (L);
\end{tikzpicture}
\par\end{centering}
\caption{Overview of the state-of-the-art for the open online dial-a-ride problem.
Bounds in parentheses were shown for the real line. Note that lower
bounds on the real line carry over to general metric spaces and the
converse is true for upper bounds. In particular, our upper bound
also holds on the line.\label{fig:state-of-art}}
\end{figure}
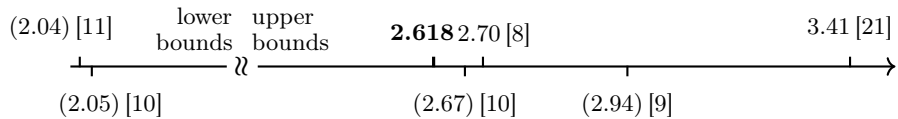

The competitive analysis of the closed online dial-a-ride problem
on general metric spaces has proven to be structurally much simpler
and conclusive results are known: The best possible competitive ratio
of 2 is achieved by the conceptually clean Smartstart algorithm, as
shown by Ascheuer et al.~\cite{AscheuerKrumkeRambau/00} and Feuerstein
and Stougie~\cite{FeuersteinStougie/01}. Ausiello et al.~\cite{Ausiello/01}
gave a matching lower bound already for TSP. The situation is more
involved on the line. Bjelde et al.~\cite{BjeldeDisserHackfeldEtal/20}
gave a sophisticated algorithm for closed online TSP on the line that
tightly matches the lower bound of~1.64 shown by Ausiello et al.~\cite{Ausiello/01}.
Birx~\cite{Birx/20} separated closed online dial-a-ride on the line
by giving a lower bound of~1.76. No better upper bound than 2 is
known in this setting, not even for preemptive algorithms. Blom et
at.~\cite{BlomKrumkePaepeStougie/01} gave a tight bound of 1.5 for
the half-line, and the best known lower bound of 1.71 for closed dial-a-ride
on the half-line is due to Ascheuer et al.~\cite{AscheuerKrumkeRambau/00}.

Clearly, most variants of the online dial-a-ride problem have resisted
tight competitive analysis for many years. As a remedy, several authors
have resorted to considering restricted classes of algorithms, restricted
adversary models, or resource augmentation. In that vein, Blom et
al.~\cite{BlomKrumkePaepeStougie/01} considered ``zealous'' (or
``diligent'') algorithms that do not stay idle if there are unserved
requests, and Birx~\cite{Birx/20} derived stronger lower bounds
for ``schedule-based'' algorithms that subdivide the execution into
schedules that may not be interrupted. Examples of restricting the
adversary include ``non-abusive'' or ``fair'' models introduced
by Krumke et al.~\cite{Krumke/02} and Blom et al.~\cite{BlomKrumkePaepeStougie/01},
that force the optimum solution to stay in the convex hull of all
released requests. In the same spirit, Hauptmeier et al.~\cite{HauptmeierKrumkeRabau/00}
adopted a ``reasonable load'' model, which requires that the length
of an optimum schedule for serving all requests revealed up to time
$t$ is bounded by a function of $t$. In terms of resource augmentation,
Allulli et al.~\cite{AllulliAusielloLaura/05} and Ausiello et al.~\cite{AusielloAllulliBonifaciLaura/06}
considered a model with ``lookahead'', where the algorithm learns
about requests before they are released. In contrast, Lipmann et
al.~\cite{Lipmann/04} considered a restricted information model
where the server learns the destination of a request only upon loading
it. Bonifaci and Stougie~\cite{BonifaciStougie/08} and Jaillet and
Wagner~\cite{JailletWagner/08} considered resource augmentation
regarding the number of servers, their speeds, and their capacities.

While we concentrate on minimizing completion time, other objectives
have been studied: Krumke~\cite{Krumke0} presented first results
for randomized algorithms minimizing expected completion time, Krumke
et al.~\cite{KrumkePaepePoensgenStougie/03} and Bienkowski et al.~\cite{BienkowskiKL/21,BienkowskiLiu/19}
minimized the sum of completion times, Krumke et al.~\cite{KrumkePaepePoensgenEtal/06,Krumke/02}
and Hauptmeier et al.~\cite{HauptmeierKrumkeRabau/00} minimized
the flow time, and Yi and Tian~\cite{YiTian/05} maximized the number
of served requests (with deadlines). Regarding other metric spaces,
Jawgal et al.~\cite{JawgalMuralidharaSrinivasan/19} considered online
TSP on a circle. Various generalizations of online dial-a-ride have
been investigated: Ausiello et al.~\cite{AusielloDemangeLauraPaschos/04}
introduced the online quota TSP, where only a minimum weighted fraction
of requests need to be served, and, similarly, Jaillet and Lu~\cite{JailletLu/11,JailletLu/14}
adopted a model where requests can be rejected for a penalty in the
objective. Jaillet and Wagner~\cite{JailletWagner/08} and Hauptmeier
et al.~\cite{HauptmeierKrumkeRambauWirth/01} allowed precedence
constraints between requests.

\section{Notation and definition of the algorithm}

Let $\sigma=(r_1, \dots, r_n)$ be a sequence of requests $r_i=(a_i,b_i;t_i)$ with release times $0<t_1 < \dots < t_n$. Note that we do not allow multiple requests to appear at the same time or a request to appear at time 0 but this is not a restriction as the release times can differ by arbitrarily small values. We let $\opt(t)$ denote the completion time of the offline optimum over all requests released not later than~$t$. A \emph{schedule} is a sequence of actions of the server, specifying when requests are collected and unloaded, how the server moves, and, in particular, when the server stays stationary. Let $\opt [t]$ denote an optimal schedule with completion time $\opt (t)$. We say that a server \emph{visits} point $p\in M$ at time~$t\geq0$ if the server is in position~$p$ at time~$t$.

The rough idea of our algorithm is to wait until we gather several requests and then start a schedule serving them. If a new request arrives during the execution of a schedule, it would be desirable to include it in the server's plan. Therefore, we check whether we can ``reset'' the server's state in a reasonable time, i.e., deliver all currently loaded requests and return to the origin, so that we can compute a new schedule. If this is not possible, we keep following the current schedule and consider the new requests later.

We introduce some notation to capture this more formally. Let $R$ be a set of requests and $x \in M$. Then, the schedule $S(R,x)$ is the shortest schedule starting from point $x$ and serving all requests in $R$. Note that this schedule can ignore the release times of the requests as we will only compute it after all requests in $R$ are released. As it is not beneficial to wait at some point during the execution of a schedule, the walked distance in $S(R,x)$ is the same as the time needed to complete it. We denote its length by $|S(R,x)|$.

Now, we can describe our algorithm. The factor $\alpha \geq 1$ will be a measure of how long we wait before starting a schedule.
A precise description of the algorithm is given below (cf. \hyperref[lazy]{Algorithm 1}). In short, whenever a new request $r=(a,b;t)$ arrives, we determine whether it is possible to serve all loaded requests and return to the origin in time $\alpha\cdot\opt(t)$. If this is possible, we do so. In this case, we say that the schedule was \emph{interrupted}. Otherwise, we ignore the request and consider it in the next schedule. Before starting a new schedule, we wait at least until time $\alpha\cdot\opt(t)$.

In the following, the algorithm $\wait(\alpha)$ with waiting parameter $\alpha\geq1$ is described. The first part of the algorithm is invoked whenever a new request $r=(a,b;t)$ is released, and the second part of the algorithm is invoked whenever the algorithm becomes idle, i.e., when the server has finished waiting or finished a schedule. We denote by $t$ the current time, by~$R_t$ the set of unserved requests at time~$t$ and by~$p_t$ the position of the server at time~$t$. There are three commands that can be executed, namely $\deliver$, $\waituntil(t')$, and $\schedule (S)$. Whenever one of these commands is invoked, the server aborts what it is currently doing and executes the new command. The command $\deliver$ instructs the server to deliver all loaded requests and return to the origin in an optimal way. The command $\waituntil(t')$ orders the server to remain at its position until time $t'$ and the command $\schedule(S)$ tells the server to execute schedule $S$. Once the server completes the execution of a command, it becomes \emph{idle}.

\SetNlSty{}{}{} 

\begin{algorithm}[h]\label{lazy}
	\caption{$\wait(\alpha)$}
	initialize: $i \gets 0$
	
	\vspace{.1cm}
	\hrule
	\vspace{.1cm}
	
	\emph{upon receiving request $r=(a,b;t)$:}\\
	\If{server can serve loaded requests and return to $\o$ until time $\alpha \cdot \opt (t)$}{
		\textbf{execute} $\deliver$ \hfill\Comment{interrupt $S^{(i)}$} 
	}
	
	\hrule
	\vspace{.1cm}
	
	\emph{upon becoming idle:}\\
	\uIf{$t<\alpha\cdot\opt(t)$}{
		\textbf{execute} $\waituntil(\alpha \cdot \opt (t))$} 
	{\ElseIf{$R_t\neq\emptyset$}{
			$i\gets i+1$, $R^{(i)}\gets R_t$, $t^{(i)}\gets t$, $p^{(i)}\gets p_t$ \\
			$S^{(i)}\gets S(R^{(i)},p^{(i)})$\\
			\textbf{execute} $\schedule ({S^{(i)})}$
		}} 
		
\end{algorithm}

\section{Analysis of $\wait$}

In this section, we analyze $\wait(\alpha)$ and show that $\wait(\alpha)$ is $1+\alpha$ competitive for $\alpha\geq\varphi=\frac{1+\sqrt{5}}{2}$. This implies in particular that $\wait(\varphi)$ is $(1+\varphi)$-competitive, i.e., that the first part of Theorem~\ref{thm:CR_of_lazy} holds.

\begin{theorem}
	For $\alpha \geq \varphi  \approx 1.618$, $\wait (\alpha)$ is $(1+\alpha)$-competitive for the open dial-a-ride problem on general metric spaces for any server capacity $c\in \N\cup\{\infty\}$.
\end{theorem}

\begin{proof}
	For a given request sequence $(r_1,\dots,r_n)$, we denote the number of schedules started by $\wait(\alpha)$ by $k\leq n$. Let $S^{(i)}$, $t^{(i)}$, $p^{(i)}$, and $R^{(i)}$ be as defined in the algorithm, i.e., $S^{(i)}$ is the $i$-th schedule started by $\wait(\alpha)$, $t^{(i)}$ is its starting time, $p^{(i)}$ its starting position, and $R^{(i)}$ is the set of requests served by~$S^{(i)}$. Observe that some schedules might be interrupted so that $R^{(1)}, \dots, R^{(k)}$ are not necessarily disjoint. Also observe that we have $p_1=\o$ and, for $i>1$, $p^{(i)}$ is either the ending position of $S^{(i-1)}$ or $\o$ if $S^{(i-1)}$ was interrupted. 
	
	We show by induction on $i$ that, for all $i \in \{1, \dots, k \}$,
	\begin{enumerate}[a)]
		\item $|S^{(i)}| \leq \opt (t^{(i)})$, and
		\item $t^{(i)} + |S^{(i)}| \leq (1+\alpha) \cdot \opt(t^{(i)})$.
	\end{enumerate}
	Note that this completes the proof since the last schedule is completed at time $t^{(k)} + |S^{(k)}|$ and since $\opt(t^{(k)})$ is the completion time of the offline optimum over all requests.
	
	Before starting the induction, let us make some observations. Since the server does not start a schedule at time $t$ if $t<\alpha\cdot \opt (t)$, we have
	\begin{equation}\label{rtime}
	t^{(i)} \geq \alpha \cdot \opt (t^{(i)})
	\end{equation}
	for all $i \in \{1, \dots, k\}$. Further, for every request $r=(a,b;t) \in R^{(i+1)} \setminus R^{(i)}$, we have ${t > t^{(i)}}$ because $R^{(i)}$ contains all unserved requests released until time $t^{(i)}$. Moreover, ${R^{(i+1)} \setminus R^{(i)} \neq \emptyset}$ because otherwise $S^{(i)}$ is not interrupted and we have $R^{(i+1)}=\emptyset$, contradicting that the algorithm starts $S^{(i+1)}$. Therefore, for all $i \in \{1, \dots, k-1\}$,
	\begin{equation}\label{optincreases}
	\opt(t^{(i+1)}) > t^{(i)} \overset{\text{\cref{rtime}}}{\geq} \alpha \cdot \opt (t^{(i)}).
	\end{equation}
	
	Now, let us start the induction.\\
	\emph{Base case:} a)  Since $\opt [t^{(1)}]$ is a schedule serving all requests in $R^{(1)}$ starting from $\o$, possibly with additional waiting times, we have 
	\begin{equation*} 
	|S^{(1)}|=|S(R^{(1)}, p^{(1)}) |=|S(R^{(1)}, \o) | \leq \opt(t^{(1)}).
	\end{equation*} 
	
	b) Consider the time $t^{(1)}$ at which schedule $S^{(1)}=S(R^{(1)},\o)$ is started, and let $t'\leq t^{(1)}$ denote the largest release time of a request in $R^{(1)}$. In particular, no requests are released in the time period $(t',t^{(1)}]$ and thus $\opt(t')=\opt(t^{(1)})$. When the request at time $t'$ is released, the server is in $\o$ so that the command $\deliver$ is completed immediately. Therefore, the server becomes idle at time~$t'$ and the waiting time is set to $\alpha \cdot \opt(t')=\alpha \cdot \opt(t^{(1)})$. Observe that $\alpha\cdot\opt(t')>\opt(t')\geq t'$. The server becomes idle again and starts schedule $S^{(1)}$ precisely at time $t^{(1)}=\alpha \cdot \opt (t^{(1)})$.
	We obtain
	\begin{equation*}
	t^{(1)}+|S^{(1)}| \overset{\text{a)}}{\leq} (1+ \alpha) \cdot \opt (t^{(1)}).
	\end{equation*}
	\emph{Induction step:} Assume that a) and b) hold for some $i\in\{1,\dots,k-1\}$. We show that this implies that a) and b) also hold for $i+1$.
	
	First, consider the case that the schedule $S^{(i)}$ is interrupted. Then, we have $p^{(i+1)}=\o$ and \mbox{$t^{(i+1)}=\alpha \cdot \opt (t^{(i+1)})$}. It immediately follows that \linebreak ${S^{(i+1)}=|S(R^{(i+1)},p^{(i+1)})| \leq \opt(t^{(i+1)})}$ because $\opt[t^{(i+1)}]$ serves all requests in $R^{(i+1)}$ (among others) and starts in $\o$. With this, we obtain
	\begin{equation*}
	t^{(i+1)} + |S^{(i+1)}| \leq (1+ \alpha) \cdot \opt (t^{(i+1)}).
	\end{equation*}
	Therefore, a) and b) hold  for $i+1$ if $S^{(i)}$ is interrupted. For the rest of the proof, assume that $S^{(i)}$ is not interrupted.
	 
	Assume that $\opt[t^{(i+1)}]$ visits $p^{(i+1)}$ before collecting any request in $R^{(i+1)}$. Then, by definition of $S^{(i+1)}=S(R^{(i+1)}, p^{(i+1)})$, we immediately see that a) holds for $i+1$ because $\opt [t^{(i+1)}]$ needs to serve all requests in $R^{(i+1)}$ after visiting point $p^{(i+1)}$. Thus, for the proof of a), it suffices to consider the case that $\opt [t^{(i+1)}]$ collects some request in $R^{(i+1)}$  before visiting $p^{(i+1)}$. We denote the first request in~$R^{(i+1)}$ collected by $\opt [t^{(i+1)}]$ by $r=(a,b;t)$. 
	
	Since $S^{(i)}$ is not interrupted, we have $R^{(i+1)} \cap R^{(i)} = \emptyset$ and thus $t> t^{(i)}$. Together with \cref{rtime}, this implies that $\opt[t^{(i+1)}]$ collects $r$ at $a$ not earlier than time $t > t^{(i)} \geq \alpha \cdot \opt (t^{(i)})$. By definition of $r$ and $S(R^{(i+1)},a)$, this implies
	\begin{equation}\label{StartAtA}
	\opt (t^{(i+1)}) \geq t + |S(R^{(i+1)},a)| > \alpha \cdot \opt (t^{(i)}) + |S(R^{(i+1)},a)|.
	\end{equation} 
	Further, since we assumed that $\opt[t^{(i+1)}]$ visits $p^{(i+1)}$ after visiting $a$ later than $\alpha \cdot \opt (t^{(i)})$ and since the server needs at least time $\dist (a,p^{(i+1)})$ to get from $a$ to $p^{(i+1)}$, we have 
	\begin{equation}\label{detour}
	\opt(t^{(i+1)}) \geq \alpha \cdot \opt (t^{(i)}) + \dist (a,p^{(i+1)}).
	\end{equation}
	Let $t_\ell\leq t^{(i+1)}$ denote the largest release time of a request in $R^{(i+1)}$. No requests appears in the, possibly empty, time interval $(t_\ell,t^{(i+1)}]$. Thus, we have ${\opt(t_\ell)=\opt(t^{(i+1)})}$. Recall that the schedule $S^{(i)}$ is not interrupted. In particular, it is not interrupted at time $t_\ell$, i.e., at time $t_\ell$, the server cannot serve all loaded requests and return to the origin until time $\alpha\cdot\opt(t_\ell)$. At time $t_\ell$, the server can trivially serve all loaded requests in time $t^{(i)}+|S^{(i)}|$ by following the current schedule, which ends in $p^{(i+1)}$. This yields
	\begin{equation}\label{nointerrupt}
	t^{(i)}+|S^{(i)}| + \dist(p^{(i+1)}, \o) 
	> \alpha \cdot \opt (t_\ell)
	= \alpha \cdot \opt (t^{(i+1)}).
	\end{equation}
	Recall that $p^{(i+1)}$ is the ending position of $S^{(i)}$ and, therefore, it is the destination of a request in~$R^{(i)}$. Since $\opt [t^{(i)}]$ needs to visit $p^{(i+1)}$ and starts in $\o$, we have $\dist(p^{(i+1)}, \o)\leq \opt (t^{(i)})$. Further, by the induction hypothesis, we have $t^{(i)} + |S^{(i)}| \leq (1+\alpha) \cdot \opt(t^{(i)})$. This yields
	\begin{align*}
	(2+\alpha) \cdot \opt (t^{(i)})
	&\geq t^{(i)}+|S^{(i)}| + \dist(p^{(i+1)}, \o)\\
	\overset{\text{\cref{nointerrupt}}}&{>} \alpha \cdot \opt (t^{(i+1)}) \\
	\overset{\text{\cref{detour}}}&{\geq} 
	\alpha ^2  \cdot \opt (t^{(i)}) + \alpha \cdot  \dist(a,p^{(i+1)}),
	\end{align*} so that
	\begin{equation}\label{distax}
	\dist (a,p^{(i+1)}) < \frac{1}{\alpha} \cdot (2+\alpha-\alpha^2) \cdot \opt (t^{(i)}).
	\end{equation}
	The schedule $S^{(i+1)}$ starts in $p^{(i+1)}$ and needs to serve all requests in $R^{(i+1)}$. By applying the triangle inequality, we can conclude that
	\begin{align*}
	|S^{(i+1)}| &\leq \dist (p^{(i+1)}, a) +  |S(R^{(i+1)},a)|\\
	\overset{\text{\cref{StartAtA}}}&{<} \dist (p^{(i+1)}, a) + \opt (t^{(i+1)})  - \alpha \cdot \opt (t^{(i)}) \\
	\overset{\text{\cref{distax}}}&{<} \left(\frac{2}{\alpha}+1-2\alpha\right) \cdot \opt (t^{(i)}) + \opt (t^{(i+1)})\\
	&\leq \opt (t^{(i+1)}),
	\end{align*}
	where the last inequality holds because we have $\left(\frac{2}{\alpha}+1-2\alpha\right) \leq 0$ if \linebreak ${\alpha\geq \frac{1+\sqrt{17}}{4}\approx 1.2808}$.
	
	It remains to show that b) also holds for $i+1$. If the schedule $S^{(i)}$ is completed before time \mbox{$\alpha \cdot \opt (t^{(i+1)})$}, the schedule $S^{(i+1)}$ is started precisely at  time  $t^{(i+1)}=\alpha \cdot \opt (t^{(i+1)})$. Together with part~a), this yields the assertion. Therefore, assume that $S^{(i)}$ is not completed before time $\alpha \cdot \opt (t^{(i+1)})$. Then, the schedule $S^{(i+1)}$ is started as soon as $S^{(i)}$ is completed. Together with the induction hypothesis, this implies $t^{(i+1)}=t^{(i)}+|S^{(i)}|\leq(1+ \alpha) \cdot \opt (t^{(i)})$. Hence, the schedule $S^{(i+1)}$ can be completed in time
	\begin{align*}
	t^{(i+1)} + |S^{(i+1)}| &\leq 
	(1+ \alpha) \cdot \opt (t^{(i)}) + |S^{(i+1)}|\\
	\overset{\text{\cref{optincreases}, a)}}&{\leq} (1+ \alpha) \cdot\frac{1}{\alpha} \cdot  \opt (t^{(i+1)}) + \opt (t^{(i+1)})\\
	& = \left(\frac{1}{\alpha}+2\right) \cdot \opt (t^{(i+1)})\\
	&\leq (1+ \alpha)\cdot \opt (t^{(i+1)}),
	\end{align*}
	where the last inequality holds because we have $\frac{1}{\alpha}+2 \leq 1+ \alpha$ if $\alpha \geq \frac{1+\sqrt{5}}{2}=\varphi$.
\end{proof}

\section{Lower bound for $\wait$}

In this section, we provide lower bounds on the competitive ratio of $\wait(\alpha)$. We give a lower bound construction for $\alpha \geq 1 $ and a separate construction for $\alpha <1$. Together they show that $\wait(\alpha)$ cannot be better than $(3/2+\sqrt{11/12})$-competitive for all $\alpha\geq0$, i.e., that Corollary~\ref{cor:global_lower_bound} holds. Furthermore, they narrow the range of parameter choices that would lead to an improvement over the competitive ratio of~$\varphi+1$.

In the following constructions, we let the metric space $(M,d)$ be the real line, i.e., $M=\R$, $\o = 0$, and $d(a,b)=|a-b|$. Note that lower bounds on the line trivially carry over to general metric spaces. Moreover, our constructions work for any given server capacity $c\in\mathbb{N}\cup\{\infty\}$ because a larger server capacity does neither change the behavior of the optimum solution nor the behavior of $\wait$.

First, observe that, for any $\alpha \geq0$, $\wait (\alpha)$ has a competitive ratio of at least $1+\alpha$. This can be easily seen by observing the request sequence consisting of the single request $r_1=(1,1;\frac{1}{2})$. In this case, the offline optimum has completed the sequence by time~1, whereas $\wait (\alpha)$ waits in $\o$ until time $\max(\alpha, \frac{1}{2})$ and then moves to 1 and serves $r_1$ not earlier than $1+\alpha$.

\begin{lemma}\label{lem:1plusalpha}
For any $\alpha\geq0$, $\wait (\alpha)$ has a competitive ratio of at least $1+\alpha$ for the open online dial-a-ride problem on the line for any capacity $c\in\mathbb{N}\cup\{\infty\}$.
\end{lemma}

Now, we give a construction for the case $\alpha \geq 1$.

\begin{proposition}\label{thm:lower_bound_alpha-geq-1}
	For $\alpha\geq1$, $\wait (\alpha)$ has a competitive ratio of at least $2+\frac{2}{3\alpha}$ for the open online dial-a-ride problem on the line.
\end{proposition}

\begin{proof}
	First, observe that, for $\alpha \geq 1/2 + \smash{\sqrt{11/12}}$, we have $2+\frac{2}{3 \alpha} \leq 1+\alpha$ so that the assertion follows from Lemma~\ref{lem:1plusalpha}. Therefore, let $\alpha \in [1,1/2 + \smash{\sqrt{11/12}})$ and let $\epsilon >0$ be small enough such that $3 \alpha +2 >3 \alpha^2+\alpha \epsilon$. Note that this is possible because $3\alpha +2 >3 \alpha ^2$ for $\alpha \in [1,(1/2)+\sqrt{11/12})$.
	
	We construct an instance of the open online dial-a-ride problem, where the competitive ratio of $\wait(\alpha)$ converges to $2+\frac{2}{3\alpha}$ for $\epsilon\rightarrow0$ (cf. Figure \ref{fig:lower_bound_large_alpha}). We define the instance by giving the requests
	\begin{equation*}
		r_1=(0,1;\epsilon) \textrm{, } r_2=(0,-1;2 \epsilon) \textrm{, and } r_3=(2-3\alpha-\epsilon,2-3\alpha-\epsilon;3\alpha+\epsilon).
	\end{equation*}
	One solution is to first serve~$r_1$ and then $r_2$. This is possible in~$3$ time units and, after this, the server is in position~$-1$. Then, the server can reach point $2-3\alpha-\epsilon$ by time $3+(3\alpha+\epsilon-2-1)=3\alpha+\epsilon$. At this point in time, $r_3$ is released and can immediately be served. Thus, we have
	\begin{equation*}
		\opt:=\opt(3\alpha+\epsilon) = 3\alpha+\epsilon.
	\end{equation*}
	
	We now analyze what $\wait(\alpha)$ does on this request sequence. We have $\opt(2 \epsilon)=3$. Thus, the server waits in~$\o$ until time $3\alpha$. Since no new request arrives until this time, the server starts an optimal schedule serving $r_1$ and $r_2$. Without loss of generality, we can assume that $\wait(\alpha)$ starts by serving~$r_2$, because the starting positions and destinations of $r_1$ and~$r_2$ are symmetrical. At time $3\alpha+\epsilon$, request~$r_3$ is released, and the server has currently loaded $r_2$. Delivering $r_2$ and returning to the origin takes the server until time $3\alpha+2$. By definition of $\alpha$ and $\epsilon$, we have
	\begin{equation}\label{noreturn}
		3\alpha+2 > 3\alpha^2+\alpha\epsilon = \alpha\opt.
	\end{equation}
	This implies that the server is not interrupted in its current schedule. It continues serving~$r_2$ and then serves~$r_1$ at time $3\alpha+3$. Together with \cref{noreturn}, it follows that, after serving $r_1$, the server immediately starts serving the remaining request $r_3$. Moving from~$1$ to $2-3\alpha-\epsilon$ takes $3\alpha-1+\epsilon$ time units, i.e., the server serves~$r_3$ at time $(3\alpha+3)+(2-3\alpha-\epsilon)=6\alpha+2+\epsilon$.
	Thus, the competitive ratio is at least
	\begin{equation*}
		\frac{6\alpha+2+\epsilon}{\opt} = \frac{6\alpha+2+\epsilon}{3\alpha+\epsilon} = 2+\frac{2-\epsilon}{3\alpha+\epsilon}.
	\end{equation*}
	The statement follows by taking the limit $\epsilon\rightarrow0$.
\end{proof}

\begin{figure}
	\begin{center}
		\begin{tikzpicture}[scale=0.95,font=\small,decoration=brace,every node/.style={scale=1}]
		\def\endpt{1pt}
		\def\reqpt{2.5pt}
		
		\def\opt{\textsc{Opt}}
		
		\def\xmax{11}
		\def\ymin{-2}
		\def\ymax{2}
		
		\def\alph{1.25}
		
		\draw[line width=1pt,black,->,>=stealth] (0,{\ymin})--(0,{\ymax});
		\draw[line width=1pt,black,->,>=stealth] ({-0.25},0)--({\xmax},0);
		\node[above left] at (0,0) {$0$};
		\node[below right, align=left] at (0,{\ymax}) {position};
		\node[above left, align=left] at ({\xmax},0) {time};
		
		\draw [fill=ForestGreen] ({\xmax-1.8},{\ymax-0.15}) rectangle ({\xmax-2},{\ymax-0.35});
		\node[below right,align=left] at ({\xmax-1.8},{\ymax}) {$\wait$};
		\draw [fill=NavyBlue] ({\xmax-1.8},{\ymax-0.55}) rectangle ({\xmax-2},{\ymax-0.75});
		\node[below right,align=left,] at ({\xmax-1.8},{\ymax-0.4}) {\opt};
		
		\draw[line width=1pt,black,-] ({3*\alph},-0.175)--++(0,0.175);
		\node[align=center] at ({3*\alph},-0.4) {$3\alpha$};
		
		\draw[line width=1pt,black,-] ({6*\alph+2},-0.075)--++(0,0.075);
		\node[align=center] at ({6*\alph+2},-0.4) {$\wait$};
		
		\draw[line width=1pt,black,-] (-0.075,{1})--++(0.075,0);
		\node[left,align=right] at (-0.075,{1}) {$1$};
		
		\draw[line width=1pt,black,-] (-0.075,-1)--++(0.075,0);
		\node[left,align=right] at (-0.075,-1) {$-1$};
		
		\draw[line width=1pt,black,-] (-0.075,{2-3*\alph-0.1})--++(0.075,0);
		\node[left,align=right] at (-0.075,{2-3*\alph-0.1}) {$2-3\alpha-\epsilon$};
		
		\draw[line width=2.5pt,black,-] (0,0)--({3*\alph},0)--({3*\alph+1},-1)--({3*\alph+3},1)--({6*\alph+2.1},{2-3*\alph-0.1});
		\filldraw[fill=ForestGreen] ({6*\alph+2.1},{2-3*\alph-0.1}) circle (\endpt);
		\draw[line width=1.5pt,ForestGreen,-] (0,0)--({3*\alph},0)--({3*\alph+1},-1)--({3*\alph+3},1)--({6*\alph+2.1},{2-3*\alph-0.1});
		
		\draw[line width=2.5pt,black,-] (0,0)--(1,1)--({3*\alph+0.1},{2-3*\alph-0.1});
		\filldraw[fill=NavyBlue] ({3\alph},{2-3*\alph}) circle (\endpt);
		\draw[line width=1.5pt,NavyBlue,-] (0,0)--(1,1)--({3*\alph+0.1},{2-3*\alph-0.1});
		
		\filldraw[fill=black] (0,0) circle (\reqpt);
		
		\draw[line width=5pt,black,-] ({3*\alph+2},0)--({3*\alph+3},1);
		\draw[line width=4pt,VioletRed,-] ({3*\alph+2},0)--({3*\alph+3},1);
		\draw[line width=2.5pt,black,-] ({3*\alph+2},0)--({3*\alph+3},1);
		\draw[line width=1.5pt,ForestGreen,-] ({3*\alph+2},0)--({3*\alph+3},1);
		\filldraw[fill=VioletRed] 	({3*\alph+2},0) circle (\reqpt);
		\filldraw[fill=VioletRed] 	({3*\alph+3},1) circle (\reqpt);
		
		\draw[line width=5pt,black,-] ({3*\alph},0)--({3*\alph+1},-1);
		\draw[line width=4pt,Goldenrod,-] ({3*\alph},0)--({3*\alph+1},-1);
		\draw[line width=2.5pt,black,-] ({3*\alph},0)--({3*\alph+1},-1);
		\draw[line width=1.5pt,ForestGreen,-] ({3*\alph},0)--({3*\alph+1},-1);
		\filldraw[fill=Goldenrod] 	({3*\alph},0) circle (\reqpt);
		\filldraw[fill=Goldenrod] 	({3*\alph+1},-1) circle (\reqpt);
		
		\filldraw[fill=Red] ({6*\alph+2.1},{2-3*\alph-0.1}) circle (\reqpt);
		
		\draw[line width=5pt,black,-] (0,0)--(1,1);
		\draw[line width=4pt,VioletRed,-] (0,0)--(1,1);
		\draw[line width=2.5pt,black,-] (0,0)--(1,1);
		\draw[line width=1.5pt,NavyBlue,-] (0,0)--(1,1);
		\filldraw[fill=VioletRed] 	(0,0) circle (\reqpt);
		\filldraw[fill=VioletRed] 	(1,1) circle (\reqpt);
		
		\draw[line width=5pt,black,-] (2,0)--(3,-1);
		\draw[line width=4pt,Goldenrod,-] (2,0)--(3,-1);
		\draw[line width=2.5pt,black,-] (2,0)--(3,-1);
		\draw[line width=1.5pt,NavyBlue,-] (2,0)--(3,-1);
		\filldraw[fill=Goldenrod] 	(2,0) circle (\reqpt);
		\filldraw[fill=Goldenrod] 	(3,-1) circle (\reqpt);
		
		\filldraw[fill=Red] ({3*\alph+0.1},{2-3*\alph-0.1}) circle (\reqpt);

		\filldraw[fill=Red] ({3*\alph+0.1},{2-3*\alph-0.1}) circle (\reqpt);
		
		\draw[line width=2.5pt,black,densely shadow,->,>=stealth] (0,-0.01)--(0,1.03);
		\draw[line width=2.5pt,black,densely shadow,->,>=stealth] (0,0.01)--(0,-1.03);
		\filldraw[fill=VioletRed] (0,0) circle (\reqpt);
		\draw[fill=Goldenrod,Goldenrod] ({-\reqpt+.5pt},0) -- ({-\reqpt+.5pt},0) arc(180:360:{\reqpt-.5pt}) --cycle;
		\draw[line width=1.5pt,VioletRed,densely dashed,->,>=stealth] (0,0)--(0,1);
		\draw[line width=1.5pt,Goldenrod,densely dashed,->,>=stealth] (0,0)--(0,-1);
		\end{tikzpicture}
	\end{center}
	\caption{Instance of the open online dial-a-ride problem on the line where $\wait(\alpha)$ has a competitive ratio of at least $2+\frac{2}{3\alpha}$ for all $\alpha\geq1$.}
	\label{fig:lower_bound_large_alpha}
\end{figure}

Next, we give a lower bound construction for $\alpha < 1$.

\begin{proposition}\label{thm:lower_bound_alpha-leq-1}
	For $\alpha\in[0,1)$, the algorithm $\wait (\alpha)$ has a competitive ratio of at least $1+\frac{3}{\alpha+1}$ for the open dial-a-ride problem on the line.
\end{proposition}

\begin{proof}
	Let $\alpha\in[0,1)$ and $\epsilon\in(0,\min\{ \frac{\alpha}{2},\frac{1}{\alpha}-\alpha,1-\alpha \})$. We construct an instance of the open dial-a-ride problem, where the competitive ratio of $\wait(\alpha)$ converges to $1+\frac{3}{\alpha+1}$ for $\epsilon\rightarrow0$ (cf. Figure \ref{fig:lower_bound_small_alpha}). We define the instance by giving the requests
	\begin{align*}
		&r_1=\left(\frac{\epsilon}{2},\frac{1}{2};\frac{\epsilon}{2}\right) \textrm{, } r_2=(1,1;\epsilon) \textrm{, } r_3=(0,0;\alpha+\epsilon) \textrm{, }\\
		 &r_4=\Bigl(\frac{1}{2}+\epsilon,1;\alpha+2\epsilon\Bigr) \textrm{, and } r_5=(1,1;\alpha+1+\epsilon).
	\end{align*}
	One solution is to first wait in $\o$ until time $\alpha+\epsilon$ and  serve~$r_3$. Then, the server can move to $\epsilon/2$, pick up~$r_1$ and deliver it. Then, we can move to~$\frac{1}{2}+\epsilon$, pick up~$r_4$ and deliver it. This can be done by time $\alpha+1+\epsilon$. Now, the server is in position~$1$ and can thus immediately serve $r_5$. It finishes serving all request in time $\alpha+1+\epsilon$. Since the last request is released at time $\alpha+1+\epsilon$, we have 
	\begin{equation}	\label{eq:optvalue}
	\opt:=\opt(\alpha+1+\epsilon)=\alpha+1+\epsilon.
	\end{equation}
	
	We now analyze what $\wait(\alpha)$ does on this request sequence. We have $\opt(\epsilon/2)=1/2$ so that $\alpha \cdot \opt(\epsilon/2)=\alpha/2> \epsilon$ and the server does not start moving before $r_2$ is released. Then, we have $\opt( \epsilon)=1$. Hence, the server waits in~$\o$ until time $\alpha$. Since no new requests arrive until this time, the server starts an optimal schedule serving $r_1$ and $r_2$, i.e., it moves to $\epsilon/2$ and picks up~$r_1$. At times $\alpha+\epsilon$ and $\alpha+2\epsilon$, $r_3$ and $r_4$ are released. We have $\opt(\alpha+\epsilon)=\opt(\alpha+2\epsilon)=\alpha+1+\epsilon$. Serving the loaded request~$r_1$ and returning to~$0$ would take the server until time
	\begin{equation}\label{eq:alpha+1_larger_alphaOpt}
		\alpha+1 \overset{\epsilon < \frac{1}{\alpha}-\alpha}{>} \alpha+(\alpha^2+\alpha\epsilon) = \alpha(\alpha+1+\epsilon) = \alpha\cdot\opt(\alpha+\epsilon).
	\end{equation}
	Thus, the server keeps following its tour and serves~$r_1$ and then~$r_2$ at time $\alpha+1$. By \cref{eq:alpha+1_larger_alphaOpt} and since $\opt(\alpha+1)=\opt(\alpha+\epsilon)$, the server immediately starts serving~$r_3$ and~$r_4$. The shortest tour is serving~$r_4$ first, i.e., the server starts moving towards $\frac{1}{2}+\epsilon$. At time $\alpha+1+\epsilon$, request~$r_5$ is released. Since
	\begin{equation*}
	\alpha+1+\epsilon \overset{\text{\cref{eq:optvalue}}}{=} \opt(\alpha+1+\epsilon)>\alpha\cdot\opt(\alpha+1+\epsilon),
	\end{equation*}
	 the server keeps following its tour, which is finished at time \linebreak ${(\alpha+1)+(1-2\epsilon)+1=3+\alpha-2\epsilon}$ in position~$\o$. Then, the server starts its last tour in order to serve~$r_5$. It moves to~$1$ and finishes serving the last request at time $4+\alpha-2\epsilon$. Thus, the competitive ratio is
	\begin{equation*}
		\frac{4+\alpha-2\epsilon}{\opt} = \frac{4+\alpha-2\epsilon}{\alpha+1+\epsilon} = 1+\frac{3-3\epsilon}{\alpha+1+\epsilon}.
	\end{equation*}
	The statement follows by taking the limit $\epsilon\rightarrow0$.
\end{proof}

\begin{figure}
	\begin{center}
		\begin{tikzpicture}[scale=0.95,font=\small,decoration=brace,every node/.style={scale=1}]
		\def\endpt{1pt}
		\def\reqpt{2.5pt}
		
		\def\opt{\textsc{Opt}}
		
		\def\xmax{12}
		\def\ymin{-1}
		\def\ymax{2.5}
		
		\def\alph{.7}
		\def\eps{.25}
		
		\draw[line width=1pt,black,->,>=stealth] (0,{\ymin})--(0,{\ymax});
		\draw[line width=1pt,black,->,>=stealth] ({-0.25},0)--({\xmax},0);
		\node[above left] at (0,0) {$0$};
		\node[below right, align=left] at (0,{\ymax}) {position};
		\node[above left, align=left] at ({\xmax},0) {time};
		
		\draw [fill=ForestGreen] ({\xmax-1.8},{\ymax-0.15}) rectangle ({\xmax-2},{\ymax-0.35});
		\node[below right,align=left] at ({\xmax-1.8},{\ymax}) {$\wait$};
		\draw [fill=NavyBlue] ({\xmax-1.8},{\ymax-0.55}) rectangle ({\xmax-2},{\ymax-0.75});
		\node[below right,align=left,] at ({\xmax-1.8},{\ymax-0.4}) {$\opt$};
		
		\draw[line width=1pt,black,-] ({2*\alph},-0.075)--++(0,0.075);
		\node[align=center] at ({2*\alph},-0.2) {$\alpha$};
		
		\draw[line width=1pt,black,-] ({2*\alph+\eps},-0.075)--++(0,0.075);
		\node[align=center] at ({2*\alph+\eps},-0.6) {$\alpha\!+\!\epsilon$};
		
		\draw[line width=1pt,black,-] ({2*\alph+8-2*\eps},-0.075)--++(0,0.075);
		\node[align=center] at ({2*\alph+8-2*\eps},-0.2) {$\wait$};
		
		\draw[line width=1pt,black,-] (-0.075,{2})--++(0.075,0);
		\node[left,align=right] at (-0.075,{2}) {$1$};
		
		\draw[line width=1pt,black,-] (-0.075,1)--++(0.075,0);
		\node[left,align=right] at (-0.075,1) {$\frac{1}{2}$};
		
		\draw[line width=2.5pt,black,-] (0,0)--({2*\alph},0)--({2*\alph+2},2)--({2*\alph+3-\eps},1+\eps)--({2*\alph+4-2*\eps},2)--({2*\alph+6-2*\eps},0)--({2*\alph+8-2*\eps},2);
		\filldraw[fill=ForestGreen] ({2*\alph+8-2*\eps},2) circle (\endpt);
		\draw[line width=1.5pt,ForestGreen,-] (0,0)--({2*\alph},0)--({2*\alph+2},2)--({2*\alph+3-\eps},1+\eps)--({2*\alph+4-2*\eps},2)--({2*\alph+6-2*\eps},0)--({2*\alph+8-2*\eps},2);
		
		\draw[line width=2.5pt,black,-] (0,0)--({2*\alph+\eps},0)--({2*\alph+2+\eps},2);
		\filldraw[fill=NavyBlue] ({2*\alph+2+\eps},2) circle (\endpt);
		\draw[line width=1.5pt,NavyBlue,-] (0,0)--({2*\alph+\eps},0)--({2*\alph+2+\eps},2);
		
		\filldraw[fill=black] (0,0) circle (\reqpt);
		
		\draw[line width=5pt,black,-] ({2*\alph},0)--({2*\alph+1},1);
		\draw[line width=4pt,VioletRed,-] ({2*\alph},0)--({2*\alph+1},1);
		\draw[line width=2.5pt,black,-] ({2*\alph},0)--({2*\alph+1},1);
		\draw[line width=1.5pt,ForestGreen,-] ({2*\alph},0)--({2*\alph+1},1);
		\filldraw[fill=VioletRed] 	({2*\alph},0) circle (\reqpt);
		\filldraw[fill=VioletRed] 	({2*\alph+1},1) circle (\reqpt);
		
		\filldraw[fill=Red] (2*\alph+2,2) circle (\reqpt);
		
		\draw[line width=5pt,black,-] ({2*\alph+3-\eps},{1+\eps})--({2*\alph+4-2*\eps},2);
		\draw[line width=4pt,Goldenrod,-] ({2*\alph+3-\eps},{1+\eps})--({2*\alph+4-2*\eps},2);
		\draw[line width=2.5pt,black,-] ({2*\alph+3-\eps},{1+\eps})--({2*\alph+4-2*\eps},2);
		\draw[line width=1.5pt,ForestGreen,-] ({2*\alph+3-\eps},{1+\eps})--({2*\alph+4-2*\eps},2);
		\filldraw[fill=Goldenrod] 	({2*\alph+3-\eps},{1+\eps}) circle (\reqpt);
		\filldraw[fill=Goldenrod] 	({2*\alph+4-2*\eps},2) circle (\reqpt);
		
		\filldraw[fill=RawSienna] ({2*\alph+6-2*\eps},0) circle (\reqpt);
		
		\filldraw[fill=Orange] ({2*\alph+8-2*\eps},2) circle (\reqpt);
		
		\draw[line width=5pt,black,-] ({2*\alph+\eps},0)--({2*\alph+1+\eps},1);
		\draw[line width=4pt,VioletRed,-] ({2*\alph+\eps},0)--({2*\alph+1+\eps},1);
		\draw[line width=2.5pt,black,-] ({2*\alph+\eps},0)--({2*\alph+1+\eps},1);
		\draw[line width=1.5pt,NavyBlue,-] ({2*\alph+\eps},0)--({2*\alph+1+\eps},1);
		\filldraw[fill=VioletRed] 	({2*\alph+\eps},0) circle (\reqpt);
		\filldraw[fill=VioletRed] 	({2*\alph+1+\eps},1) circle (\reqpt);
		
		\filldraw[fill=Red] (2*\alph+2+\eps,2) circle (\reqpt);
		
		\filldraw[fill=RawSienna] (2*\alph+\eps,0) circle (\reqpt);
		
		\draw[line width=5pt,black,-] ({2*\alph+1+2*\eps},{1+\eps})--({2*\alph+2+\eps},2);
		\draw[line width=4pt,Goldenrod,-] ({2*\alph+1+2*\eps},{1+\eps})--({2*\alph+2+\eps},2);
		\draw[line width=2.5pt,black,-] ({2*\alph+1+2*\eps},{1+\eps})--({2*\alph+2+\eps},2);
		\draw[line width=1.5pt,NavyBlue,-] ({2*\alph+1+2*\eps},{1+\eps})--({2*\alph+2+\eps},2);
		\filldraw[fill=Goldenrod] 	({2*\alph+1+2*\eps},{1+\eps}) circle (\reqpt);
		\filldraw[fill=Goldenrod] 	({2*\alph+2+\eps},2) circle (\reqpt);
		
		\filldraw[fill=Orange] (2*\alph+2\eps,2) circle (\reqpt);

		\draw[line width=2.5pt,black,densely shadow,->,>=stealth] (0,-0.01)--(0,1.03);
		\filldraw[fill=VioletRed] (0,0) circle (\reqpt);
		\draw[line width=1.5pt,VioletRed,densely dashed,->,>=stealth] (0,0)--(0,1);
		
		\filldraw[fill=Red] (0,2) circle (\reqpt);
		
		\filldraw[fill=RawSienna] (2*\alph+\eps,0) circle (\reqpt);
		
		\draw[line width=2.5pt,black,densely shadow,->,>=stealth] (2*\alph+\eps,0.99+\eps)--(2*\alph+\eps,2.03);
		\filldraw[fill=Goldenrod] (2*\alph+\eps,1+\eps) circle (\reqpt);
		\draw[line width=1.5pt,Goldenrod,densely dashed,->,>=stealth] (2*\alph+\eps,1+\eps)--(2*\alph+\eps,2);
		
		\filldraw[fill=Orange] (2*\alph+2\eps,2) circle (\reqpt);
		\end{tikzpicture}
	\end{center}
	\caption{Instance of the open online dial-a-ride problem on the line where $\wait(\alpha)$ has a competitive ratio of at least $1+\frac{3}{\alpha+1}$ for all $\alpha\in[0,1)$.}
	\label{fig:lower_bound_small_alpha}
\end{figure}
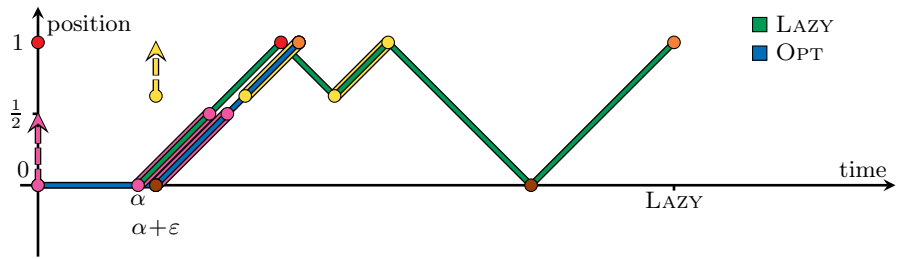

We now give another lower bound construction which is stonger than the previous one for large~$\alpha < 1$.

\begin{proposition}\label{thm:lower_bound_large_alpha-leq-1}
	For $\alpha\in[0,1)$, the algorithm $\wait (\alpha)$ has a competitive ratio of at least $2+\alpha+\frac{1-\alpha}{2+3\alpha}$ for the open online dial-a-ride problem on the line.
\end{proposition}

\begin{proof}
	Let $\epsilon>0$ be small enough. We construct an instance of the open online  dial-a-ride problem on the line, where the competitive ratio of $\wait(\alpha)$ converges to $2+\alpha+\frac{1-\alpha}{2+3\alpha}$ for $\epsilon\rightarrow0$ (cf. Figure~\ref{fig:lower_bound_large_alpha_smaller_1}). We distinguish between three cases.
	
	\emph{Case 1} ($\alpha\in[0,2/3)$): We define the instance by giving the requests
	\begin{eqnarray*}
		r_1 &=& (0,1;\epsilon), \\
		r_2 &=& (-\alpha,-\alpha;\alpha+\epsilon), \\
		r_3 &=& (2+\alpha-\epsilon,2+\alpha-\epsilon;\alpha+2\epsilon), \\
		r_4 &=& (2+\alpha-\epsilon,2+\alpha-\epsilon;2+3\alpha).
	\end{eqnarray*}
	One solution is to move to~$-\alpha$ and wait there until time~$\alpha+\epsilon$. Then,~$r_2$ can be served, and the server can move to~$0$ where it arrives at time~$2\alpha+\epsilon$. It picks up~$r_1$ and delivers it at time~$2\alpha+1+\epsilon$ at~$1$. It keeps moving to position~$2+\alpha-\epsilon$ and serves~$r_3$ and~$r_4$ there at time~$2+3\alpha$. Since the last request is released at time~$2+3\alpha$, we have
	\begin{equation*}
	\opt:=\opt(2+3\alpha)=2+3\alpha.
	\end{equation*}
	We now analyze what $\wait(\alpha)$ does. We have $\opt(\epsilon)=1+\epsilon$. Thus, the server starts waiting in~$0$ until time $\alpha(1+\epsilon)$. At time~$\alpha(1+\epsilon)$, the server starts an optimal schedule over all unserved requests, i.e., over~$\{r_1\}$. It picks up~$r_1$ and starts moving towards~$1$. At time~$\alpha+\epsilon>\alpha(1+\epsilon)$, request~$r_2$ arrives. We have $\opt(\alpha+\epsilon)=2\alpha+1+\epsilon$. Serving the loaded request~$r_1$ and returning to the origin would take the server until time
	\begin{equation*}
	\alpha(1+\epsilon)+2 > \alpha(2\alpha+1+\epsilon)=\alpha\opt(\alpha+\epsilon),
	\end{equation*}
	i.e., the server continues its current schedule. At time~$\alpha+2\epsilon$,~$r_3$ is released. We have $\opt(\alpha+2\epsilon)=2+3\alpha$. Serving the loaded request and returning to the origin would still take until time
	\begin{equation*}
	\alpha(1+\epsilon)+2 > 2\alpha+3\alpha^2=\alpha\opt(\alpha+2\epsilon),
	\end{equation*}
	where the inequality follows from the fact that $\alpha<\frac{2}{3}$. Thus, the server continues the current schedule, which is finished at time $\alpha(1+\epsilon)+1$ in position~$1$. 
	The server waits until time \mbox{$\max\{1+\alpha,\alpha\opt(\alpha+2\epsilon)\}$} and then starts the next schedule. Since $\alpha\opt(\alpha+2\epsilon)=2\alpha+3\alpha^2<2+3\alpha$, the next schedule is thus started before $r_4$ is released. It is faster to serve $r_3$ before $r_2$  in this schedule because the server starts from point 1. At time~$2+3\alpha$, request~$r_4$ is released. Since this does not change the completion time of the optimum and because the server started the current schedule not earlier than time $\alpha\opt(\alpha+2\epsilon)$, it continues the current schedule. The second schedule takes $(1+\alpha-\epsilon)+(2+2\alpha-\epsilon)=3+3\alpha-2\epsilon$ time units and ends in~$-\alpha$. The last schedule, in which~$r_4$ is served, is started immediately and takes~$2+2\alpha-\epsilon$ time units. Hence, $\wait(\alpha)$ takes at least
	\begin{equation*}
	(\alpha(2+3\alpha))+(3+3\alpha-2\epsilon)+(2+2\alpha-\epsilon)=5+7\alpha+3\alpha^2-3\epsilon
	\end{equation*}
	time units to serve all requests.
	
	\emph{Case 2} ($\alpha\in[\frac{2}{3},\frac{\sqrt{37-12\epsilon}-1}{6})$): We define the instance by giving the requests
	\begin{eqnarray*}
		r_1 &=& (0,1;\epsilon), \\
		r_2 &=& (-\alpha,-\alpha;\alpha+\epsilon), \\
		r_3 &=& \Bigl(\frac{2}{\alpha}+1-2\alpha-\epsilon,\frac{2}{\alpha}+1-2\alpha-\epsilon;\alpha+2\epsilon\Bigr), \\
		r_4 &=& (2+\alpha-\epsilon,2+\alpha-\epsilon;2+\alpha-\epsilon), \\
		r_5 &=& (2+\alpha-\epsilon,2+\alpha-\epsilon;2+3\alpha).
	\end{eqnarray*}
	One solution is to move to~$-\alpha$ and wait there until time~$\alpha+\epsilon$. Then,~$r_2$ can be served, and the server can move to~$0$ where it arrives at time~$2\alpha+\epsilon$. It picks up~$r_1$ and delivers it at time~$2\alpha+1+\epsilon$ at~$1$. It keeps moving to position~$\frac{2}{\alpha}+1-2\alpha-\epsilon$, where it arrives at time~$\frac{2}{\alpha}+1$ and immediately serves~$r_3$. It continues to move to~$2+\alpha-\epsilon$ and serves~$r_4$ and~$r_5$ there at time~$2+3\alpha$. Since the last request is released at time~$2+3\alpha$, we have
	\begin{equation*}
	\opt=\opt(2+3\alpha)=2+3\alpha.
	\end{equation*}
	We now analyze what $\wait(\alpha)$ does. We have $\opt(\epsilon)=1+\epsilon$. Thus, the server starts waiting in~$0$ until time $\alpha(1+\epsilon)$. At time~$\alpha(1+\epsilon)$, the server starts an optimal schedule over all unserved requests, i.e., over~$\{r_1\}$. It picks up~$r_1$ and starts moving towards~$1$. At time~$\alpha+\epsilon>\alpha(1+\epsilon)$, request~$r_2$ arrives. We have $\opt(\alpha+\epsilon)=2\alpha+1+\epsilon$. Serving the loaded request~$r_1$ and returning to the origin would take the server until time
	\begin{equation*}
	\alpha(1+\epsilon)+2 > \alpha(2\alpha+1+\epsilon)=\alpha\opt(\alpha+\epsilon),
	\end{equation*}
	i.e., the server continues its current schedule. At time~$\alpha+2\epsilon$,~$r_3$ is released. We have $\opt(\alpha+2\epsilon)=\frac{2}{\alpha}+1$. Serving the loaded request and returning to the origin would still take until time
	\begin{equation*}
	\alpha(1+\epsilon)+2 > 2+\alpha = \alpha\opt(\alpha+2\epsilon).
	\end{equation*}
	Thus, the server continues the current schedule which is finished at time $1+\alpha$ in position~$1$. Since $1+\alpha<2+\alpha=\alpha\opt(\alpha+2\epsilon)$, the server starts waiting in~$1$ until time $2+\alpha$. At time $2+\alpha-\epsilon$, request~$r_4$ is released. We have ${\opt(2+\alpha-\epsilon)=2+3\alpha}$. It would take the server until time
	\begin{equation*}
	3+\alpha-\epsilon > 2\alpha+3\alpha^2 = \alpha\opt(2+\alpha-\epsilon)
	\end{equation*}
	to return to the origin, where the inequality follows from the fact that \linebreak ${\alpha<\frac{\sqrt{37-12\epsilon}-1}{6}}$. Thus, the server starts waiting until time \linebreak ${\alpha\opt(2+\alpha-\epsilon)=2\alpha+3\alpha^2}$. After it finished waiting, it starts the second schedule and tries to serve~$r_3$,~$r_4$ and then~$r_2$. At time~$2+3\alpha$, request~$r_5$ is released. Since this does not change the completion time of the optimum and because the server started the current schedule at time $\alpha\opt(\alpha+2\epsilon)$, it continues its current schedule. The second schedule takes $(1+\alpha-\epsilon)+(2+2\alpha-\epsilon)=3+3\alpha-2\epsilon$ time units and ends in~$-\alpha$. The last schedule, in which~$r_5$ is served, is started immediately and takes~$2+2\alpha-\epsilon$ time units. Hence, $\wait(\alpha)$ takes until time
	\begin{equation*}
	(\alpha(2+3\alpha))+(3+3\alpha-2\epsilon)+(2+2\alpha-\epsilon)=5+7\alpha+3\alpha^2-3\epsilon
	\end{equation*}
	to serve all requests.
	
	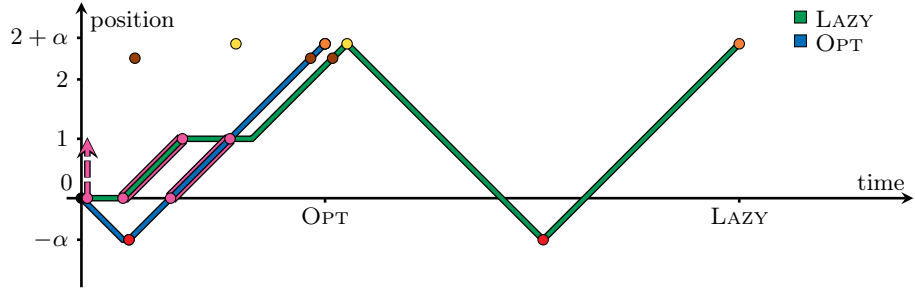
\begin{figure}
		\begin{center}
			\begin{tikzpicture}[scale=0.79,font=\small,decoration=brace,every node/.style={scale=1}]
				\def\endpt{1pt}
				\def\reqpt{2.5pt}
				
				\def\opt{\textsc{Opt}}
				
				\def\xmax{14}
				\def\ymin{-1.5}
				\def\ymax{3.3}
				
				\def\alph{.7}
				\def\eps{.1}
				
				\draw[line width=1pt,black,->,>=stealth] (0,{\ymin})--(0,{\ymax});
				\draw[line width=1pt,black,->,>=stealth] ({-0.25},0)--({\xmax},0);
				\node[above left] at (0,0) {$0$};
				\node[below right, align=left] at (0,{\ymax}) {position};
				\node[above left, align=left] at ({\xmax},0) {time};
				
				\draw [fill=ForestGreen] ({\xmax-1.8},{\ymax-0.15}) rectangle ({\xmax-2},{\ymax-0.35});
				\node[below right,align=left] at ({\xmax-1.8},{\ymax}) {$\wait$};
				\draw [fill=NavyBlue] ({\xmax-1.8},{\ymax-0.55}) rectangle ({\xmax-2},{\ymax-0.75});
				\node[below right,align=left,] at ({\xmax-1.8},{\ymax-0.4}) {$\opt$};
				
				\draw[line width=1pt,black,-] ({2+3*\alph},-0.075)--++(0,0.075);
				\node[align=center] at ({2+3*\alph},-0.3) {$\opt$};
				
				\draw[line width=1pt,black,-] ({5+7*\alph+3*\alph*\alph-3*\eps},-0.075)--++(0,0.075);
				\node[align=center] at ({5+7*\alph+3*\alph*\alph-3*\eps},-0.3) {$\wait$};
				
				\draw[line width=1pt,black,-] (-0.075,{2+\alph})--++(0.075,0);
				\node[left,align=right] at (-0.075,{2+\alph}) {$2+\alpha$};
				
				\draw[line width=1pt,black,-] (-0.075,{2})--++(0.075,0);
				\node[left,align=right] at (-0.075,{2}) {$2$};
				
				\draw[line width=1pt,black,-] (-0.075,1)--++(0.075,0);
				\node[left,align=right] at (-0.075,1) {$1$};
				
				\draw[line width=1pt,black,-] (-0.075,{-\alph})--++(0.075,0);
				\node[left,align=right] at (-0.075,{-\alph}) {$-\alpha$};
				
				\draw[line width=2.5pt,black,-] (0,0)--(\alph,0)--(1+\alph,1)--(2*\alph+3*\alph*\alph,1)--(1+3*\alph+3*\alph*\alph-\eps,2+\alph-\eps)--(3+5*\alph+3*\alph*\alph-2*\eps,-\alph)--(5+7*\alph+3*\alph*\alph-3*\eps,2+\alph-\eps);
				\filldraw[fill=ForestGreen] (5+7*\alph+3*\alph*\alph-3*\eps,2+\alph-\eps) circle (\endpt);
				\draw[line width=1.5pt,ForestGreen,-] (0,0)--(\alph,0)--(1+\alph,1)--(2*\alph+3*\alph*\alph,1)--(1+3*\alph+3*\alph*\alph-\eps,2+\alph-\eps)--(3+5*\alph+3*\alph*\alph-2*\eps,-\alph)--(5+7*\alph+3*\alph*\alph-3*\eps,2+\alph-\eps);
				
				\draw[line width=2.5pt,black,-] (0,0)--(\alph,-\alph)--(\alph+\eps,-\alph)--(2+3*\alph,2+\alph-\eps);
				\filldraw[fill=NavyBlue] (2+3*\alph,2+\alph-\eps) circle (\endpt);
				\draw[line width=1.5pt,NavyBlue,-] (0,0)--(\alph,-\alph)--(\alph+\eps,-\alph)--(2+3*\alph,2+\alph-\eps);
				
				\filldraw[fill=black] (0,0) circle (\reqpt);
				
				\draw[line width=5pt,black,-] ({\alph},0)--({\alph+1},1);
				\draw[line width=4pt,VioletRed,-] ({\alph},0)--({\alph+1},1);
				\draw[line width=2.5pt,black,-] ({\alph},0)--({\alph+1},1);
				\draw[line width=1.5pt,ForestGreen,-] ({\alph},0)--({\alph+1},1);
				\filldraw[fill=VioletRed] 	({\alph},0) circle (\reqpt);
				\filldraw[fill=VioletRed] 	({\alph+1},1) circle (\reqpt);
				
				\filldraw[fill=Red] (3+5*\alph+3*\alph*\alph-2*\eps,-\alph) circle (\reqpt);
				
				\filldraw[fill=RawSienna] (3*\alph*\alph+2/\alph-\eps,2/\alph+1-2*\alph-\eps) circle (\reqpt);
				
				\filldraw[fill=Goldenrod] (1+3*\alph+3*\alph*\alph-\eps,2+\alph-\eps) circle (\reqpt);
				
				\filldraw[fill=Orange] (5+7*\alph+3*\alph*\alph-3*\eps,2+\alph-\eps) circle (\reqpt);
				
				\draw[line width=5pt,black,-] ({2*\alph+\eps},0)--({2*\alph+1+\eps},1);
				\draw[line width=4pt,VioletRed,-] ({2*\alph+\eps},0)--({2*\alph+1+\eps},1);
				\draw[line width=2.5pt,black,-] ({2*\alph+\eps},0)--({2*\alph+1+\eps},1);
				\draw[line width=1.5pt,NavyBlue,-] ({2*\alph+\eps},0)--({2*\alph+1+\eps},1);
				\filldraw[fill=VioletRed] 	({2*\alph+\eps},0) circle (\reqpt);
				\filldraw[fill=VioletRed] 	({2*\alph+1+\eps},1) circle (\reqpt);
				
				\filldraw[fill=Red] (\alph+\eps,-\alph) circle (\reqpt);
				
				\filldraw[fill=RawSienna] (2/\alph+1,2/\alph+1-2*\alph-\eps) circle (\reqpt);
				
				\filldraw[fill=Goldenrod] (2+3*\alph,2+\alph-\eps) circle (\reqpt);
				
				\filldraw[fill=Orange] (2+3*\alph,2+\alph-\eps) circle (\reqpt);

				\draw[line width=2.5pt,black,densely shadow,->,>=stealth] (\eps,-0.01)--(\eps,1.03);
				\filldraw[fill=VioletRed] (\eps,0) circle (\reqpt);
				\draw[line width=1.5pt,VioletRed,densely dashed,->,>=stealth] (\eps,0)--(\eps,1);
				
				\filldraw[fill=Red] (\alph+\eps,-\alph) circle (\reqpt);
				
				\filldraw[fill=RawSienna] (\alph+2*\eps,2/\alph+1-2*\alph-\eps) circle (\reqpt);
				
				\filldraw[fill=Goldenrod] (2+\alph-\eps,2+\alph-\eps) circle (\reqpt);
				
				\filldraw[fill=Orange] (2+3*\alph,2+\alph-\eps) circle (\reqpt);
			\end{tikzpicture}
		\end{center}
		\caption{Instance of the open online dial-a-ride problem on the line where $\wait(\alpha)$ has a competitive ratio of at least $2+\alpha\frac{1-\alpha}{2+3\alpha}$ for all $\alpha\in[0,1)$. This is the construction of Case 2 in the proof of Theorem~\ref{thm:lower_bound_large_alpha-leq-1}.}
		\label{fig:lower_bound_large_alpha_smaller_1}
	\end{figure}
	
	\emph{Case 3} ($\alpha\in[\frac{\sqrt{37-12\epsilon}-1}{6},1)$): We define the instance by giving the requests
	\begin{eqnarray*}
		r_1 &=& (0,1;\epsilon), \\
		r_2 &=& (-\alpha,-\alpha;\alpha+\epsilon), \\
		r_3 &=& \Bigl(\frac{2}{\alpha}+1-2\alpha-\epsilon,\frac{2}{\alpha}+1-2\alpha-\epsilon;\alpha+2\epsilon\Bigr), \\
		r_4 &=& \Bigl(\frac{3}{\alpha}+1-2\alpha-\frac{(2+\alpha)\epsilon}{\alpha},\frac{3}{\alpha}+1-2\alpha-\frac{(2+\alpha)\epsilon}{\alpha};2+\alpha-\epsilon\Bigr), \\
		r_5 &=& (2+\alpha-\epsilon,2+\alpha-\epsilon;3+\alpha-3\epsilon), \\
		r_6 &=& (2+\alpha-\epsilon,2+\alpha-\epsilon;2+3\alpha).
	\end{eqnarray*}
	One solution is to move to~$-\alpha$ and wait there until time~$\alpha+\epsilon$. Then,~$r_2$ can be served, and the server can move to~$0$ where it arrives at time~$2\alpha+\epsilon$. It picks up~$r_1$ and delivers it at time~$2\alpha+1+\epsilon$ at~$1$. It keeps moving to position~$\frac{2}{\alpha}+1-2\alpha-\epsilon$, where it arrives at time~$\frac{2}{\alpha}+1$ and immediately serves~$r_3$. It continues to move to~$\frac{3}{\alpha}+1-2\alpha-\frac{(2+\alpha)\epsilon}{\alpha}$, where it arrives at time~$\frac{3}{\alpha}+1-\frac{2\epsilon}{\alpha}$ and immediately serves~$r_4$. Lastly, it moves to~$2+\alpha-\epsilon$, where it arrives at time~$2+3\alpha$ and serves~$r_4$ and~$r_5$ there. Since the last request is released at time~$2+3\alpha$, we have
	\begin{equation*}
	\opt=\opt(2+3\alpha)=2+3\alpha.
	\end{equation*}
	We now analyze what $\wait(\alpha)$ does. We have $\opt(\epsilon)=1+\epsilon$. Thus, the server starts waiting in~$0$ until time $\alpha(1+\epsilon)$. At time~$\alpha(1+\epsilon)$, the server starts an optimal schedule over all unserved requests, i.e., over~$\{r_1\}$. It picks up~$r_1$ and starts moving towards~$1$. At time~$\alpha+\epsilon>\alpha(1+\epsilon)$, request~$r_2$ arrives. We have $\opt(\alpha+\epsilon)=2\alpha+1+\epsilon$. Serving the loaded request~$r_1$ and returning to the origin would take the server until time
	\begin{equation*}
	\alpha(1+\epsilon)+2 > \alpha(2\alpha+1+\epsilon)=\alpha\opt(\alpha+\epsilon),
	\end{equation*}
	i.e., the server continues its current schedule. At time~$\alpha+2\epsilon$, request~$r_3$ is released. Now, we have \mbox{$\opt(\alpha+2\epsilon)=\frac{2}{\alpha}+1$}. Serving the loaded request and returning to the origin would still take until time
	\begin{equation*}
	\alpha(1+\epsilon)+2 > 2+\alpha = \alpha\opt(\alpha+2\epsilon).
	\end{equation*}
	Thus, the server continues the current schedule which is finished at time $1+\alpha$ in position~$1$. Since $1+\alpha<2+\alpha=\alpha\opt(\alpha+2\epsilon)$, the server starts waiting in~$1$ until time $2+\alpha$. At time $2+\alpha-\epsilon$, request~$r_4$ is released. We have ${\opt(2+\alpha-\epsilon)=\frac{3}{\alpha}+1-\frac{2\epsilon}{\alpha}}$. It would take the server until time
	\begin{equation*}
	3+\alpha-\epsilon > 3+\alpha-2\epsilon = \alpha\opt(2+\alpha-\epsilon),
	\end{equation*}
	to return to the origin, i.e., the server starts waiting until time \linebreak ${\alpha\opt(2+\alpha-\epsilon)=3+\alpha-2\epsilon}$. At time $3+\alpha-3\epsilon$, request~$r_5$ is released. We have $\opt(3+\alpha-3\epsilon)=2+3\alpha$. It would take the server until time
	\begin{equation*}
	4+\alpha-3\epsilon > 2\alpha+3\alpha^2 = \alpha\opt(3+\alpha-3\epsilon)
	\end{equation*}
	to return to the origin, where the inequality follows from the fact that $\alpha<1$ and that $\epsilon$ is small. Thus, the server waits in~$1$ until time \linebreak ${\alpha\opt(3+\alpha-3\epsilon)=2\alpha+3\alpha^2}$ and then starts its second schedule to serve requests $r_3$,~$r_4$,~$r_5$ and then~$r_2$. At time~$2+3\alpha$, request~$r_6$ is released. Since this does not change the completion time of the optimum and because the server started the current schedule at time \mbox{$\alpha\opt(3+\alpha-3\epsilon)$}, it continues its current schedule. The second schedule takes $(1+\alpha-\epsilon)+(2+2\alpha-\epsilon)=3+3\alpha-2\epsilon$ time units and ends in~$-\alpha$. The last schedule, in which~$r_6$ is served, is started immediately and takes~$2+2\alpha-\epsilon$ time units. Hence, $\wait(\alpha)$ takes until time
	\begin{equation*}
	(\alpha(2+3\alpha))+(3+3\alpha-2\epsilon)+(2+2\alpha-\epsilon)=5+7\alpha+3\alpha^2-3\epsilon
	\end{equation*}
	to serve all requests.
	
	In all three cases, the optimal solution is $2+3\alpha$ and the algorithm takes at least $5+7\alpha+3\alpha^2-3\epsilon$ time units. Thus, the competitive ratio of $\wait(\alpha)$ is at least
	\begin{equation*}
	\frac{5+7\alpha+3\alpha^2-3\epsilon}{2+3\alpha} = 2+\alpha+\frac{1-\alpha-3\epsilon}{2+3\alpha}.
	\end{equation*}
	The statement follows by taking the limit $\epsilon\rightarrow0$.
\end{proof}

Now, we combine our results for the lower bounds (cf. Figure~\ref{fig:plot_lower_bounds}). Combining~\cref{lem:1plusalpha} and \mbox{\cref{thm:lower_bound_alpha-geq-1}}, we obtain that, for $\alpha \geq 1$, $\wait(\alpha)$ has a competitive ratio of at least $\max\{1+\alpha,2 + 2/3 \alpha \}$, which proves the lower bound of \cref{thm:CR_of_lazy}.
Minimizing over $\alpha \geq 1$ yields a competitive ratio of at least~$\smash{3/2+\sqrt{11/12}} > 2.457$ in that domain.
For the case $\alpha <1$, we have seen in \cref{thm:lower_bound_alpha-leq-1} that the algorithm $\wait (\alpha)$ has a competitive ratio of at least $1 + 3/(\alpha +1) > 5/2$. 
Together, this proves \cref{cor:global_lower_bound}.

\begin{figure}
	\begin{center}
		\begin{tikzpicture}
		\begin{axis}[
		samples=100,
		xmin=0,
		xmax=3,
		xtick={0,1,2,...,4},
		ytick={2,3,4},
		ymin=2,
		ymax=4,
		xlabel=waiting parameter $\alpha$,
		ylabel=competitive ratio $\rho$,
		height=8cm,
		width=12cm,
		ylabel near ticks,
		xlabel near ticks,
		]
		
		\addplot[no markers, thick, domain=0:1.457][black!40!green,line width=1pt, dashed]{x+1};
		\addplot[no markers, thick, domain=1.457:1.618][black!40!green,line width=1pt]{x+1};
		\addplot[no markers, thick, domain=1.618:4][name path=B1,black!40!green,line width=1pt]{x+1};
		
		\addplot[no markers, domain=1:1.079][name path=B2,black!40!red,line width=1pt]{2+2/3/x};
		\addplot[no markers, domain=1.079:1.457][black!40!red,line width=1pt]{2+2/3/x};
		\addplot[no markers, domain=1.457:4][black!40!red,line width=1pt, dashed]{2+2/3/x};
		
		\addplot[no markers, domain=0:0.695][black!40!blue,line width=1pt,dashed]{2+x+(1-x)/(2+3*x)};
		\addplot[no markers, domain=0.695:1][name path=B3,black!40!blue,line width=1pt]{2+x+(1-x)/(2+3*x)};
		
		\addplot[no markers, domain=0:0.695][name path=B4,black!20!orange,line width=1pt]{1+3/(1+x)};
		\addplot[no markers, domain=0.695:1][black!20!orange,line width=1pt,dashed]{1+3/(1+x)};
		
		\draw[densely dashed,line width=0.5pt,black!70] (axis cs:1.618,\pgfkeysvalueof{/pgfplots/ymin}) -- (axis cs:1.618,\pgfkeysvalueof{/pgfplots/ymax});
		\addplot[no markers,densely dashed, domain=0:3,line width=0.5pt, black!70]{2.618};
	  \addplot[no markers,densely dotted, domain=0:3,line width=0.5pt, black]{2.457};
		
		\addplot[no markers, domain=1.618:4][name path=A1,black,line width=.5pt]{4};
		\addplot[no markers, domain=1:1.079][name path=A2,black,line width=.5pt]{4};
		\addplot[no markers, domain=0.695:1][name path=A3,black,line width=.5pt]{4};
		\addplot[no markers, domain=0:0.695][name path=A4,black,line width=.5pt]{4};
		
		\addplot[black!6!red!15, postaction={pattern=my north east lines,pattern color=gray}] fill between[of=A1 and B1];
		\addplot[black!6!red!15, postaction={pattern=my north east lines,pattern color=gray}] fill between[of=A2 and B2];
		\addplot[black!6!red!15, postaction={pattern=my north east lines,pattern color=gray}] fill between[of=A3 and B3];
		\addplot[black!6!red!15, postaction={pattern=my north east lines,pattern color=gray}] fill between[of=A4 and B4];
		
		\node[above] at (axis cs: 1.8,2.05) {$\alpha=\varphi$};
		\node[above right] at (axis cs: 2.5,2.618) {$\rho= \varphi+1$};
		\node[below right] at (axis cs: 2.5,2.457) {$\rho= 2.457$};
		\end{axis}
		\end{tikzpicture}
	\end{center}
	\caption{Lower bounds on the competitive ratio of $\wait(\alpha)$ depending on~$\alpha$. The lower bound of \cref{lem:1plusalpha} is depicted in green, the lower bound of \cref{thm:lower_bound_alpha-geq-1} in red, the lower bound of \cref{thm:lower_bound_alpha-leq-1} in orange, and the lower bound of \cref{thm:lower_bound_large_alpha-leq-1} in blue.
	The highlighted area over the plot indicates the domain of $\alpha$ for which no improvement over $1+\varphi$ is possible.}
	\label{fig:plot_lower_bounds}
\end{figure}
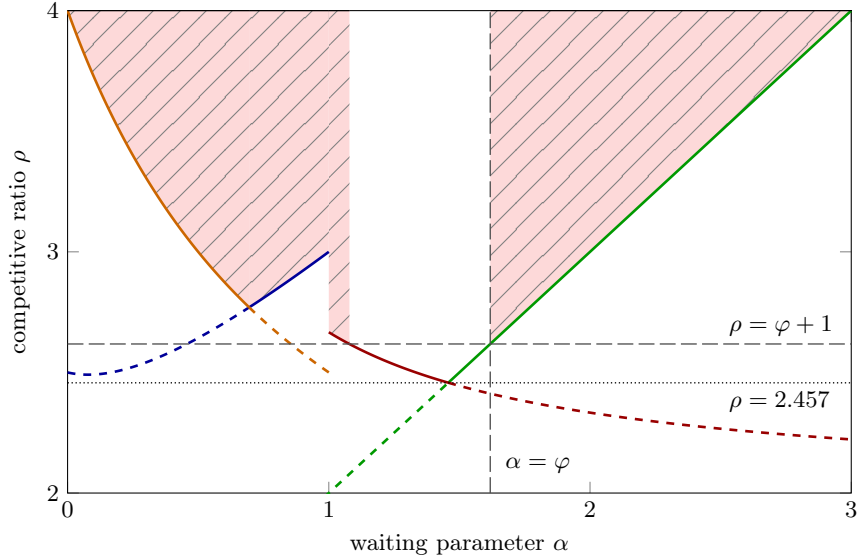

Moreover, we conclude that the results above narrow the range for $\alpha$ in which $\wait (\alpha)$ might have a competitive ratio better than $\varphi+1$. 
By \cref{lem:1plusalpha}, it follows that $\wait (\alpha)$ cannot have a better competitive ratio than $\varphi+1$ for any $\alpha >\varphi \approx 1.618$.
By \cref{thm:lower_bound_alpha-geq-1}, we obtain that $\wait (\alpha)$ has competitive ratio at least $\varphi+1$ for any $\alpha $ with $1 \leq \alpha \leq \smash{\frac{2\varphi}{3}}\approx 1.079$. \cref{thm:lower_bound_alpha-leq-1} yields that, for~\mbox{$0\leq\alpha\leq0.695$}, the competitive ratio of $\wait(\alpha)$ is at least~$2.768>\varphi+1$. Lastly, for $0.695<\alpha<1$, \cref{thm:lower_bound_large_alpha-leq-1} gives a lower bound of~$2.768$ on the competitive ratio of $\wait(\alpha)$. To summarize, an improvement of the competitive ratio of $\wait (\alpha)$ might only be possible for some \mbox{$\alpha \in (2\varphi/3,\varphi)\approx[1.08,1.618)$}, which proves \cref{cor:phi+1_lower_bound}.

%
%
%
 \bibliographystyle{splncs04}
\bibliography{Lazy}
%
%
%
%
%
\end{document}